\documentclass{elsarticle}
\makeatletter
\makeatletter
\def\doi#1{\gdef\@doi{#1}}\def\@doi{}

\usepackage[utf8]{inputenc}
\usepackage{stmaryrd}
\usepackage{amsmath, amssymb, mathrsfs}
\usepackage{hyperref}
\usepackage{comment}

\usepackage[ruled, vlined]{algorithm2e}
\usepackage{algorithmic}

\usepackage[english]{babel}

\begin{document}

\renewcommand{\floatpagefraction}{.9}

\newtheorem{theo}{Theorem}[section]
\newtheorem{lem}[theo]{Lemma}
\newtheorem{prop}[theo]{Proposition}
\newtheorem{cor}[theo]{Corollary}
\newtheorem{conj}{Conjecture}
\newtheorem{quest}[theo]{Question}
\newtheorem{rem}[theo]{Remark}
\newtheorem{ex}[theo]{Example}
\newdefinition{deftn}{Definition}
\newdefinition{rmk}{Remark}
\newproof{proof}{Proof}

\title{Matrix-F5 algorithms over finite-precision complete discrete valuation fields}

\author[tv]{Tristan Vaccon}
\ead{tristan.vaccon@univ-rennes1.fr}
\address[tv]{Université de Rennes 1}

\begin{abstract}
Let $(f_1,\dots, f_s) \in \mathbb{Q}_p [X_1,\dots, X_n]^s$ be a sequence of homogeneous polynomials with $p$-adic coefficients. Such system may happen, for example, in arithmetic geometry. Yet, since $\mathbb{Q}_p$ is not an effective field, classical algorithm does not apply.

We provide a definition for an approximate Gröbner basis with respect to a monomial order $w.$ We design a strategy to compute such a basis, when precision is enough and under the assumption that the input sequence is regular and the ideals $\langle f_1,\dots,f_i \rangle$ are weakly-$w$-ideals. The conjecture of Moreno-Socias states that for the grevlex ordering, such sequences are generic.

Two variants of that strategy are available, depending on whether one lean more on precision or time-complexity. For the analysis of these algorithms, we study the loss of precision of the Gauss row-echelon algorithm, and apply it to an adapted Matrix-F5 algorithm. Numerical examples are provided.

Moreover, the fact that under such hypotheses, Gröbner bases can be computed stably has many applications. Firstly, the mapping sending $(f_1,\dots,f_s)$ to the reduced Gröbner basis of the ideal they span is differentiable, and its differential can be given explicitly. Secondly, these hypotheses allows to perform lifting on the Grobner bases, from $\mathbb{Z}/p^k \mathbb{Z}$ to $\mathbb{Z}/p^{k+k'} \mathbb{Z}$ or $\mathbb{Z}.$ 

Finally, asking for the same hypotheses on the highest-degree homogeneous components of the entry polynomials allows to extend our strategy to the affine case.

\end{abstract}


\maketitle

\section{Introduction}

Ideal study and polynomial system solving are crucial problem in computer algebra, with numerous applications, either theoretical (as in algebraic geometry) or in applied mathematics (as in cryptography). To that intent, Gröbner bases computation is a decisive tool.

A classical strategy to tackle a problem over $\mathbb{Q}$ consists in reducing it modulo many different primes and then recombine the solutions. In that case, one can choose freely the primes and discard those that lead to inefficient or irrelevant computations. This method applies also to Gröbner bases and leads to the notion of "lucky" primes.
Nevertheless, the advent of arithmetic geometry has seen the emergence of questions that are purely local (\textit{i.e.} where the prime $p$ is fixed at the very beginning and one can not vary it). As an example, one can cite the recent work of Caruso and Lubicz \cite{Caruso:2014} who gave an algorithm to compute lattices in some $p$-adic Galois representations. A related question is the study of $p$-adic deformation spaces of Galois representations. Since the work of Taylor and Wiles \cite{Taylor:1995}, we know that these spaces play a crucial role in many questions in number theory. Being able to compute such spaces appears then as an interesting question of algorithmics and require the use of purely $p$-adic Gröbner bases. Yet, no practical survey of Gröbner bases over $p$-adic fields are actually available. This motivates our study.

In this document, we present a way to deal with Gröbner bases for ideals of $\mathbb{Q}_p[X_1, \dots, X_n]$ and  $\mathbb{F}_q((t))[X_1, \dots, X_n]$ with a strong assumption on their structure that assure numerical stability.
In that case, we provide a matrix-F5 algorithm to compute an approximate Gröbner bases of such an ideal, while being able to certify the leading monomials of the ideal. This also  proves the differentiability of the computation of reduced Gröbner bases under these assumption, and enables new lifting methods.

\bigskip

\noindent \textbf{Related works.}

In the last few decades, the need for approximate Gröbner bases for computation over floating-point numbers has risen many studies. Sasaki and Kako provide in \cite{Sasaki:2007} \cite{Sasaki:2010} a wonderful introduction to this topic, by classifying the cancellation that might happen when handling floating-point number.
Shirayanagi \& Sweedler \cite{Shirayanagi:1998}, Kondratyev, Stetter \& Winkler \cite{Kondratyev:2004}, Nagasaka \cite{Nagasaka:2009}, Stetter \cite{Stetter:2005}, Traverso \& Zanoni \cite{Traverso:2002}, Faugère \& Liang \cite{Faugere:2010} and many more have contributed to this topic. Yet, their point of view was always that of floating-point, whose behavior is not identical to that of $\mathbb{Q}_p$ or formal series.

Meanwhile, a $p$-adic approach to Gröbner bases over $\mathbb{Q}$ has been studied by Winkler \cite{Winkler:1988}, Pauer \cite{Pauer:1992}, Gräbe \cite{Grabe:1993}, Arnold \cite{Arnold:2003}, and Renault and Yokoyama \cite{Renault:2006}. Yet, their works all have in common that they are interested in a global problem over $\mathbb{Q}$ or $\mathbb{Z}$ and they choose one or many $p$'s adapted to their problem. Hence, their ideas can not directly apply to the handling of a problem over $\mathbb{Q}_p$ or $\mathbb{F}_p((t))$ where $p$ can not vary.

A shorter version of this article has been published in the Proceedings of the 39th International Symposium on Symbolic and Algebraic Computation (ISSAC 2014). It contains Section 2 and 3 and some material of Section 4 and 5.

\bigskip

\noindent \textbf{Main results.}
For $K= \mathbb{Q}_p$  or $\mathbb{F}_p((t)),$ and $R=\mathbb{Z}_p$ or $\mathbb{F}_p[[t]]$ respectively, polynomials in $K[X_1,\dots, X_n]$ or $R[X_1,\dots, X_n]$  can only be handled with finite precision over the coefficients. Let $w$ be a monomial ordering and let $f=(f_1, \dots, f_s) \in R^s$ be homogeneous polynomials satisfying the two structure hypotheses: \begin{itemize}
\item \textbf{H1}: $(f_1,\dots,f_s)$ is a regular sequence.
\item \textbf{H2}: the $\left\langle f_1, \dots, f_i \right\rangle$ are weakly-$w$-ideals (see Definition \ref{weak w}).
\end{itemize}
These hypotheses ensure some regularity property: in a neighborhood of $f$ satisfying \textbf{H1} and \textbf{H2}, the application mapping a sequence to its reduced Gröbner basis is differentiable (and continuous), and we compute the differential explicitly in Theorem \ref{Diff}. Hence, around such an $f,$ one can safely work with approximations. On the opposite, if \textbf{H1} or \textbf{H2} is relaxed, the continuity is no longer guaranteed (see  Section \ref{Continuity}), which means that the computation may not be achieved with approximated inputs. More precisely, under our structure hypotheses, we exhibit an explicit precision \[prec_{MF5}(\left\lbrace f_1, \dots, f_s \right\rbrace, D, w),\] essentially given by minors of the Macaulay matrices defined by $f,$ such that approximations of $f$ up to $prec_{MF5}$ determine well-defined approximation of Gröbner bases, compatible with the precision and with unambiguous leading terms. We provide in Definition \ref{App GB} a suitable notion of approximate Gröbner bases regarding to finite-precision coefficients. We define an approximate $D$-Gröbner basis accordingly.
To compute such $D$-Gröbner bases, we define in Algorithm \ref{algo wMF5} the weak Matrix-F5 algorithm, with the following result:
\begin{theo} \label{thm wMF5}
Let $(f_1,\dots,f_s) \in K[X_1,\dots, X_n]^s$ be homogeneous polynomials satisfying  \textbf{H1} and \textbf{H2}. Let $(f_1',\dots,f_s')$ be approximations of the $f_i$'s with precision $m$ on the coefficients.
Then, if $m$ is large enough, an approximate $D$-Gröbner basis of $(f_1',\dots,f_s')$ regarding to $w$ is well-defined. It can be computed by the weak Matrix-F5 algorithm.

Moreover, if the $f_i$'s are in $R[X_1,\dots,X_n],$ then $m \geq prec_{MF5}$ is enough, and the loss in precision is upper-bounded by $prec_{MF5}.$

The complexity is in $O \left( sD \binom{n+D-1}{D}^3 \right)$ operations in $R$ at precision $m,$ as $D \rightarrow +\infty.$
\end{theo}

We remark that the conjecture of Moreno-Socias implies that sequences satisfying \textbf{H1} and \textbf{H2} for the grevlex ordering are generic. We also remark that operations in $R = \mathbb{Z}_p$ or $\mathbb{F}_p[[t]]$ at precision $m$ can be computed, by usual algorithms, in $\widetilde{O}(m \log p)$ binary operations.

We explain in Section \ref{num exp} why $prec_{MF5}$ is not sharp, along with numerical examples.

If one lean more on precision than time-complexity, we show in Theorem \ref{thm wM} that, under the assumptions \textbf{H1} and \textbf{H2} and the $f_i$'s in $R[X_1,\dots,X_n],$, we can drop the F5 criterion in order to obtain a smaller sufficient precision for an approximate Gröbner basis to be computed: $prec_{Mac},$ see Definition \ref{prec Mac}. Time-complexity is then in $O \left( s^2 D\binom{n+D-1}{D}^3 \right)$ operations in $R$ at precision $m,$  as $D \rightarrow +\infty.$

Furthermore, the special hypotheses \textbf{H1} and \textbf{H2} allow lifting of Gröbner bases: given $G$ an approximate Gröbner basis of $\left\langle F \right\rangle$ and some $m,k$ and $M$ such that $(G+O(p^k))=(F+O(p^m)) \cdot (M + O(p^m)),$ one can compute in $O \left( (s+\sharp G ) \ \binom{n+D-1}{D}^2 \right)$ operations at precision $m+l$ an approximate Gröbner basis of $F+O(p^{m+l}).$ $\sharp G $ stands for the cardinal of $G.$ This implies that to compute the reduced Gröbner basis $G$ of some $F =(f_1,\dots,f_s) \in \mathbb{Q}[X_1,\dots,X_n]$ satisfying \textbf{H1} and \textbf{H2}, it is enough to perform the computation of an approximate Gröbner basis at some high-enough entry precision $m$ and lift it to $\mathbb{Q}[X_1,\dots,X_n].$ The total complexity is then in $O \left( s^2 D \binom{n+D-1}{D}^3 \right)$ operations in at precision $m$ and $O \left( (s+ \sharp G  ) \binom{n+D-1}{D}^2 \right)$ operations in $\mathbb{Q}.$ In other words, the cost of the linear algebra is then borne by computation at finite precision.

Finally, even though all the previous results were presented for homogeneous entry polynomials, they can be extended by replacing the hypotheses \textbf{H1} and \textbf{H2} on the entry polynomials by the same hypotheses on their homogeneous components of highest degree.

\bigskip

\noindent \textbf{Structure of the paper.}
In Section 2, we explain the setting of our paper: finite-precision complete discrete-valuation fields, and analyze the Gaussian row-echelon algorithm when performed over such fields. Section 3 applies this analysis to the study of the Matrix-F5 algorithm. We then provide and analyze a weak Matrix-F5 algorithm, and a variant for precision-efficiency.
Section 4 addresses the computation of the differential of reduced Gröbner bases computation and the optimality of the \textbf{H1} and \textbf{H2} hypotheses. Section 5 provides some experimental examples. 
Section 6 shows how one can lift efficiently under these hypotheses, and finally, Section 7 addresses the affine case.

\section{Finite-precision CDVF and row-echelon form computation}

The objective of this Section is first to introduce finite-precision complete discrete-valuation fields. We study the behavior of the precision when performing elementary operations, and from it, derive an analysis of the loss in precision when performing Gaussian row-echelon form computation 

\subsection{Setting}

Throughout this paper, $K$ is a field with a discrete valuation $val$ such that $K$ is complete with respect to the norm defined by $val$. We denote by $R=O_K$ its ring of integers, $m_K$ its maximal ideal and $k=O_K/m_K$ its fraction field. We denote by CDVF (complete discrete-valuation field) such a field. We refer to Serre's Local Fields \cite{Serre:1979} for an introduction to such fields. Let $\pi \in R$ be a uniformizer for $K$ and let $S_K \subset R$ be a system of representatives of $k=O_K/m_K.$ All numbers of $K$ can be written uniquely under its $\pi$-adic power series development form: $\sum_{k \geq l} a_k \pi^l$ for some $ l \in \mathbb{Z}$, $a_k \in S_K$.

The case that we are interested in is when $K$ might not be an effective field, but $k$ is (\textit{i.e.} there are constructive procedures for performing rational operations in $k$ and for deciding whether or not two elements in $k$ are equal). Symbolic computation can then be performed on truncation of $\pi$-adic power series development.
We will denote by finite-precision CDVF such a field, and finite-precision CDVR for its ring of integers. Classical examples of such CDVF are $K = \mathbb{Q}_p$, with $p$-adic valuation, and $\mathbb{Q}[[X]]$ or $\mathbb{F}_q[[X]]$ with $X$-adic valuation.
We assume that $K$ is such a finite-precision CDVF.

Let $A = K[X_1,\dots, X_n]$, and $w$ a monomial order on $A$. Let $B=R[X_1,\dots, X_n]$. We denote by $A_d$ the degree-$d$ homogeneous polynomials of $A,$ $A_{\leq d}$ the polynomials of total degree less than $d$, and when $u=(u_1,\dots,u_n) \in \mathbb{Z}_{\geq 0}^n$, we write $X^u$ for $X_1^{u_1} \dots X_n^{u_n}.$ If $P \in A$ is an homogeneous polynomials, we denote by $\vert P \vert$ its degree.

\subsection{Precision over $R$ and its behavior}

Elements of $K$ can be symbolically handled only up to a truncation of their $\pi$-adic power series development. Therefore, we manipulate quantities of the form $\sum_{i=k}^{d-1} a_i \pi^i+O(\pi^d)$, where $O(\pi^d)$ denotes $\pi^d R$. 

\begin{deftn}
To study the precision on an approximation of a number in $K,$ we define the \textbf{order} (or absolute precision) of $x=\sum_{i=k}^{d-1} a_i \pi^i+O(\pi^d)$ to be $d$.
\end{deftn}

The number of significant digits of $x$ would be a much more involved but as least as interesting object to study. 

We can track the behavior of the order when performing elementary operations. For this, let $n_0 < m_0, $ $n_1  < m_0$ be integers, and $\varepsilon=\sum_{j=0}^{m_0-n_0-1} a_j \pi^j$, $\mu =\sum_{j=0}^{m_1-n_1-1} b_j \pi^j$, with $a_j,b_j \in S_K$, and $a_0,b_0 \neq 0.$ It is then well-known that, 
\begin{align*}
(\varepsilon \pi^{n_0}+O(\pi^{m_0}))+(\mu \pi^{n_1}+O(\pi^{m_1})) &=\varepsilon \pi^{n_0}+\mu \pi^{n_1} \\
& \hphantom{=} +O(\pi^{min(m_0,m_1)}), \\
\end{align*}
and consequently, the addition of two number know up to order $n$ is known up to order $n.$
Similar formulae exist for all elementary operations. We only use the following:

\begin{prop}[division]
\label{div}
\[ \frac{\varepsilon \pi^{n_1}+O(\pi^{m_1})}{\mu \pi^{n_0}+O(\pi^{m_0})}=
\varepsilon \mu^{-1} \pi^{n_1-n_0}+O(\pi^{\min (m_1-n_0,m_0+n_1-2n_0)}).\]
\end{prop}

As a consequence, we can already see why finite-precision CDVF have a very different behavior than floating-point numbers: if $a=x+O(\pi^n)$ and $b = y + O(\pi^n)$ are elements of $K$ known up to the order $n$, then $a+b=(x+y)+O(\pi^n)$ is known up to the order $n$. Because of round-off errors, this does not happen with floating-point numbers. 

\subsection{The Gaussian row-echelon form algorithm}

We now apply Lemma \ref{div} to the study of Gaussian row-echelon form computation.
We first begin with recalling in Algorithm \ref{gaussian elim} what we mean with row-echelon form and Gaussian elimination.

\begin{deftn}
Let $M$ be an $n \times m$ matrix. Let $ind_{M} : \: \left\lbrace 1, r \right\rbrace \rightarrow \mathbb{Z}_{\leq 0} \cup \left\lbrace \infty \right\rbrace$ map $i$ to the index of the columns of the first non-zero entry on the $i$-th row of $M.$ Then $M$ is said to be under row-echelon form if the index function is strictly increasing.

$M$ is said to be under row-echelon form up to permutation if there exists $P$ a permutation matrix such that $PM$ is under row-echelon form.
\end{deftn}

\IncMargin{1em}
\begin{algorithm} \label{gaussian elim}
 \DontPrintSemicolon

 \caption{The Gaussian elimination algorithm}
 \SetKwInOut{Input}{input}\SetKwInOut{Output}{output}

 \Input{$M$, an $n \times m$ matrix.}
 \Output{$\widetilde{M}$, a row-echelon form of $M,$ up to permutation. }

\BlankLine

\Begin{
$\widetilde{M} \leftarrow M$;  \;
\eIf{$n_{col}=1$ or $n_{row}=0$ or $M$ has no non-zero entry}{
				Return $\widetilde{M}$; \;
	}{			
Find the coefficient $M_{i,1}$ on the first column with the smallest valuation; \;
Swap rows to put it in first row; \;
By pivoting with the first row, eliminate the coefficients of the other rows on the first column; \;
Proceed recursively on the submatrix $\widetilde{M}_{i \geq 2, j \geq 2}$; \;
Return $\widetilde{M}$; \;}}

\end{algorithm}
\DecMargin{1em}

We emphasis that when eliminating coefficients with the pivot, we produce real zeros, and not some $O(\pi^k).$
Otherwise, the resulting matrix is not under row-echelon form (up to permutation).

\subsection{How to pivot}

We now make clear how one can pivot and eliminate coefficients.

\begin{prop}[Pivoting] \label{piv}
Let $n_0 \leq n_1 <n $ be integers, and $\varepsilon=\sum_{j=0}^{n-n_1-1} a_j \pi^j$, $\mu =\sum_{j=0}^{n-n_0-1} b_j \pi^j$, with $a_j,b_j \in S_K$, and $a_0,b_0 \neq 0.$

To put a \lq\lq{}real zero\rq\rq{} on the coefficient $M_{i,j}= \varepsilon \pi^{n_1}+O(\pi^n)$, we eliminate it with a pivot $piv=\mu \pi^{n_0}+O(\pi^n)$ on row $L$. This can be performed by the following operation on the $i$-th row $L_i$:

\[L_i\leftarrow L_i-\frac{M_{i,j}}{piv}L=L_i+(\varepsilon \mu^{-1} \pi^{n_1-n_0}+O(\pi^{n-n_0}))L,\]
along with the symbolic operation $M_{i,j}\leftarrow 0.$
\end{prop}
\begin{proof}
The symbolic operations $M_{i,j}\leftarrow 0$ is just a part of the symbolic operation $L_i\leftarrow L_i-\frac{M_{i,j}}{piv}L$.
Yet, for any other coefficient of $L_i$, symbolic computation is not relevant and what is performed is $L_i+(\varepsilon \mu^{-1} \pi^{n_1-n_0}+O(\pi^{n-n_0}))L$.

Indeed, we prove that $\frac{M_{i,j}}{piv}=\varepsilon \mu^{-1} \pi^{n_1-n_0}+O(\pi^{n-n_0}).$

This is a direct consequence of Proposition \ref{div}: $\frac{M_{i,j}}{piv}=\frac{\varepsilon \pi^{n_1}+O(\pi^n)}{\mu \pi^{n_0}+O(\pi^n)}$, and therefore $\frac{M_{i,j}}{piv}=\varepsilon \mu^{-1} \pi^{n_1-n_0}+O(\pi^{n-n_0}),$ since $\min (n+n_1-2 n_0,n-n_0,n+n-2*n_0)=n-n_0$.

\end{proof}

\subsection{Gaussian row-echelon form computation}
We are now able to track the loss of precision when performing Algorithm \ref{gaussian elim} to compute a row echelon form of a matrix. The result is the following:

\begin{theo} \label{Prec-row}
Let M  be a matrix $n \times m$ ($0 \leq n \leq m$) with coefficients in $R$ all known with absolute precision $k \geq 0$ and such that its principal minor $\Delta=det((M_{i,j})_{1\leq i \leq n, 1 \leq j \leq n})$ satisfies $val(\Delta) < k$.

Then, the maximum loss of absolute precision while performing Gaussian row-reduction on M can be upper-bounded by $val(\Delta)$.
\end{theo}
\begin{proof}

To prove this result, we first study the pivoting process, and then conclude by induction on the number of rows.

When performing the Gauss row-echelon form computation, we first look for the coefficient $M_{i,1}$ on the first column with the smallest valuation. Then, we put it (via permutation of rows) to the first row, $L_1$.We denote it by $piv$, and let $n_1$ be its valuation.

As in Proposition \ref{piv}, we then pivot all the rows $L_i$ below the first one: $L_i\leftarrow L_i-\frac{M_{i,j}}{piv}L_1,$ and if we denote by $M^{(1)}$ the resulting matrix, then the coefficients of the sub-matrix $M^{(1)}_{2 \leqslant i \leqslant n, \: 2  \leqslant j \leqslant m}$ are known up to $O(\pi^{n-n_1})$.

We then proceed recursively with the pivoting process on this sub-matrix.

We first remark that the result is clear for matrices with only $n=1$ rows. We also remark that in the previous pivoting process, the operations performed on the rows change the principal minor only up to a sign, and we have  $\Delta = \pm piv \times \det M^{(1)}_{2 \leqslant i \leqslant n, \: 2  \leqslant j \leqslant n}.$

Then the result is clear by induction on $n$ the number of columns.
\end{proof}

\subsection{A more refined result}

In the following section, we will apply this result on row-echelon form computation to study the computation of Gröbner bases, but beforehand, a more sophisticated result is available if one consider matrices with possibly more rows than columns:

\begin{prop}\label{refined}
Let $M$ be a matrix $n \times m$ ($0 \leq n, m$) with coefficients in $R$ all known with absolute precision $k \geq 0$. Let $l \leq m$ be such that there is an $l$-minor on the $l$ first columns $C_1, \dots, C_l,$ with valuation strictly less than $k$.

Let $\Delta$ be the product of the pivots of the Gaussian row-reduction of $M$ up to column $l$.

Then, the maximum loss of absolute precision while performing Gaussian row-reduction on M, up to the $l$-th column, can be upper-bounded by $val(\Delta)$ and moreover, $val( \Delta)$ attains the smallest valuation of an $l$-minor on $(C_1, \dots, C_l).$ 
\end{prop}
\begin{proof}
This comes from the following fact: in the ring of integers of a complete discrete valuation ring, an ideal $I$ is generated by any of its element whose valuation attains $\min \left( val(I) \right).$

Here, if we define $I_{\text{minor}}$ to be the ideal in $R$ generated by the $l$-minors on $(C_1, \dots, C_l)$, then $I_{\text{minor}}$ remains unchanged by any of the operations on the row of $M$ performed during the Gaussian row-echelon form computation (the matrices of these operations all are invertible over $R$).

Once the row-echelon form computation is completed, there is only one non-zero $l$-minor on  $(\widetilde{C}_1, \dots, \widetilde{C}_l)$ (the $l$ first columns of $\widetilde{M}$),  and its value is the product $\Delta$ of the pivots chosen during the computation.

Therefore, $\Delta$ generates $I_{\text{minor}}$, and attains $\min \left( val(I) \right).$
\end{proof}

As a consequence, Gaussian reduction on such a matrix $M$ up to column $l$ provide the choice of pivots which yields the smallest upper-bound of the loss in precision.

\section{Matrix F5 algorithm and precision issues}
In this section, we show that our analysis of the loss in precision during Gauss reduction can be applied to understand how we can compute Gröbner bases of some ideals in $K[X_1, \dots, X_n]$. Our main tool will be Faugère\rq{}s Matrix F5 algorithm in a slightly modified version, and to that intent, we first describe the idea of the Matrix F5 algorithm, in a general setting.

\subsection{Matrix-F5}

The main reference concerning the Matrix-F5 algorithm are Bardet's PhD thesis \cite{Bardet:2004}, Bardet, Faugère and Salvy's analysis of the complexity of the F5 algorithm in \cite{Bardet:2014} and Eder and Faugère's survey of F5 algorithms in \cite{Eder:2014}.
We first recall some basic facts about matrix-algorithm to compute Gröbner bases, and present the Matrix-F5 algorithm.

\begin{deftn}
Let $B_{n,d}=(x^{d_i})_{1 \leq i \leq \binom{n-1}{n+d-1}} $ be the monomials of $A_d$, ordered decreasingly regarding to $w$.
Then for $f_1, \dots, f_s \in A$ homogeneous polynomials, with $\vert f_i \vert = d_i$, and $d \in \mathbb{N}$, we define $Mac_d(f_1, \dots, f_s)$ to be the following matrix: 
\begin{eqnarray*}
&& \begin{array}{ccccc} x^{d_1} > & \dots &> & \dots & >x^{d_{\binom{n+d-1}{n-1}}} \end{array}  \\
\begin{array}{c}
x^{\alpha_{1,1}}f_1 \\
\vdots \\
x^{\alpha_{1,\binom{n+d-d_1-1}{n-1}}} \\
x^{\alpha_{2,1}} f_2 \\
\vdots \\
x^{\alpha_{s,\binom{n+d-d_s-1}{n-1}}} f_s \\
\end{array} && \begin{bmatrix}
\hphantom{x^{d_1} >} & \hphantom{ \dots } & \hphantom{ > } & \hphantom{ \dots } & \hphantom{>x^{d_{\binom{n-1}{n+d-1}}} } \\ 
~ & ~ & ~ & ~ & ~ \\ 
~ & ~ & ~ & ~ & ~ \\ 
\hphantom{x^{d_1} >} & \hphantom{x^\alpha} & \text{\Huge *} &    & ~ \\ 
~ & ~ & ~ & ~ & ~ \\ 
~ & ~ & ~ & ~ & ~ \\ 
\end{bmatrix} 
\end{eqnarray*}

with $x^{ \alpha_{i,j}} \in B_{n,d-d_i}$. The rows of the matrix $Mac_d(f_1, \dots, f_s)$ are the polynomials $x^{ \alpha_{i,j}} f_i$ written in the basis $B_{n,d}$ of $A_d$.

We note that $Im(Mac_d(f_1, \dots, f_s))=I \cap A_d$, and the first non-zero coefficient of a row of $Mac_d(f_1,\dots, f_i)$ is the leading coefficient of the corresponding polynomial.
\end{deftn}

\begin{theo}[Lazard \cite{Lazard:1983}] \label{echelon}
For an homogeneous ideal $I=(f_1,\dots , f_s) $, $f_1,\dots , f_s$ is a Gröbner basis of $I$ if and only if: for all $d \in \mathbb{N}$, $Mac_d(f_1, \dots, f_s)$ contains an \textit{echelon basis} of $Im(Mac_d(f_1, \dots, f_s))$.
\end{theo}

By \textit{echelon basis}, we mean the following
\begin{deftn}
Let $g_1,\dots, g_r$ be homogeneous polynomials of degree $d$. Let $M$ be the matrix whose $i$-th row is the row vector corresponding to $g_i$ written in $B_{n,d}$.
Then we say that $g_1,\dots, g_r$ is an \textit{echelon basis} of $Im(M)$ if there is a permutation matrix $P$ such that $PM$ is under row-echelon form.
\end{deftn}

From this theorem, it is easy to derive an algorithm to compute Gröbner bases: compute the row-echelon form of all the $Mac_d(f_1,\dots, f_s)$, for varying $d$. Faugère's F5 criterion provides a decisive improvement with a way to remove most of the unnecessary computation. One can look at Faugère's original article \cite{Faugere:2002} or to the survey of Eder and Faugère \cite{Eder:2014} for an introduction to the F5 criterion, but it can be summed up in the following theorem.

\begin{theo}[F5 criterion] \label{discard}
If $i \in \llbracket 2,s \rrbracket$, and if we discard all the rows $x^\alpha f_j$ of $Mac_d(f_1, \dots, f_i)$ such that $x^\alpha \in LM(I_{j-1})$, for all $j \in \llbracket 2,i \rrbracket$, and if we denote by $\overline{Mac_d(f_1, \dots, f_i)}$ this matrix, $Im(\overline{Mac_d(f_1, \dots, f_i)})=Im(Mac_d(f_1, \dots, f_i)).$ A reduction to zero of a row by elementary operations over the rows leads to a syzygy that is not principal.

If $(f_1, \dots, f_s)$ is a regular sequence, then $\overline{Mac_d(f_1, \dots, f_i)}$ is injective, and no reduction to zero can be performed.
\end{theo}

Theorem \ref{discard} yields the matrix-F5 algorithm: compute the row-echelon form of the $\overline{Mac_d(f_1, \dots, f_i)}$ sequentially in $d$ and $i$, with the computation of $\overline{Mac_{d-d_i}(f_1, \dots, f_{i-1})}$ being enough to apply the F5-criterion on $Mac_d(f_1, \dots, f_i)$ to provide $\overline{Mac_d(f_1, \dots, f_i)}$.
One can build a Gröbner basis of $I$ by adding to the $f_i$'s the polynomials corresponding to the rows of the row-echelon form of $\overline{Mac_d(f_1, \dots, f_i)}$ which provide new leading monomial.

Nevertheless, there is no criterion on up to what $d$ Macaulay matrices should be echelonized. This is why we define $D$-Gröbner basis (see for example \cite{Faugere:2011}):

\begin{deftn}
Let $I$ be an ideal of $A$, $w$ a monomial ordering on $A$ and $D$ an integer. Then $(g_1, \dots , g_l)$ is a $D$-Gröbner basis of $I$ if for any $f \in I$, homogeneous of degree less than $D$, there exists $1\leq i \leq l$ such that, regarding to $w$, $LM(g_i)$ divides $LM(f)$.
\end{deftn}

Thus, if we perform the Matrix-F5 algorithm up to the degree $D$ Macaulay matrix $Mac_D(f_1, \dots, f_s)$, we obtain a $D$-Gröbner basis.

\subsection{Precision issues}

We now try to understand what happens when the entries are known only up to finite precision. To that intent, we give a definition of what we expect an approximate Gröbner basis to be:

\begin{deftn} \label{App GB}
Let $f_i+\sum_{\vert u \vert = d_i} O(\pi^{n_{u,i}}) X^u$, $1 \leq i \leq s$ , be approximations of homogeneous polynomials in $A$. The $n_{u,i}$ belong to $\mathbb{Z}_{\geq 0} \cup \left\lbrace + \infty \right\rbrace.$
Then an approximate Gröbner basis, regarding to the monomial ordering $w$, of the ideal generated by these polynomials is a finite sequence $(g_i+\sum_{\vert u \vert = \vert g_i \vert } O(\pi^{m_{u,i}}))$, $m_{u,i}  \in  \mathbb{Z}_{\geq 0} \cup \left\lbrace + \infty \right\rbrace,$ of approximation of polynomials such that: for any $a_{u,i} \in \pi^{n_{u,i}} R$, there exists some $b_{u,i} \in \pi^{m_{u,i}} R$ such that the $g_i + \sum_{\vert u \vert = \vert g_i \vert } b_{u,i} $'s form a Gröbner basis, regarding to $w$, of the ideal generated by the  $f_i+\sum_{\vert u \vert = d_i} a_{u,i} X^u$'s.
Moreover, we require that if $X^u$ is a monomial of degree $\vert g_i \vert $ such that $X^u >_w LM(g_i)$, then $m_{u,i}=+\infty$ (and the coefficient of $X^u$ in $g_i$ is zero). In other words, we require that the leading monomials of the $g_i$ does not depend on the approximation.
\end{deftn}

As seen in the previous section, if the input polynomials form a regular sequence, then all matrices considered in the F5 algorithm are injective. Yet, this is not enough in order to be able to certify that we get an approximate Gröbner bases.

For example, the injective matrix, 
\[\begin{bmatrix}
1+O(\pi^{10}) & 1+O(\pi^{10}) & 1+O(\pi^{10}) & 0 \\
1+O(\pi^{10}) & 1+O(\pi^{10}) & 1+O(\pi^{10}) & 1+O(\pi^{10}) \\
\end{bmatrix},
\]
become, after the first step in the computation of its row-echelon form.

\[\begin{bmatrix}
1+O(\pi^{10}) & 0 & 1+O(\pi^{10}) & 0 \\
0 & O(\pi^{10}) & O(\pi^{10}) & 1+O(\pi^{10}) \\
\end{bmatrix},
\]

Yet, there is no way, with only operations on the rows, to decide whether the coefficient of index $(2,2)$ is the first non-zero coefficient of the second row or if it is the one of index $(2,3)$ or $(2,4).$ Thus, we can not know what is the row-echelon form of the matrix nor the leading monomials of polynomials corresponding to its rows.

Hence, the $(f_1,\dots,f_s)$ such that the Matrix-F5 algorithm could give a satisfactory answer must have a special shape: when performing the row-echelon computation on the Macaulay matrices, no column without pivot is encountered. With $w$ being our monomial ordering, an ideal $\left\langle f_1,\dots, f_s \right\rangle$ such that every $Mac_{d}(f_1,\dots, f_j)$ satisfies this property is called a $w$-ideal. 

$w$-ideal have been heavily studied in the field of generic initial ideal: for example, in \cite{Conca:2005}, Conca and Sidman proved that the generic initial ideal (\textit{i.e.} with a generic change of variable) of a generic set of points in $\mathbb{P}^r$ is a $w$-ideal.

Yet, this is not the generic case as, for example, Pardue showed in \cite{Pardue:2010} that the ideal generated by 6 quadrics in 6 variables is not a grevlex-ideal.

Fortunately, we can study a somehow weaker condition: weakly-$w$-ideal.

\begin{deftn} \label{weak w}
Let $I$ be an ideal in $A$, and $w$ be a monomial order on $A$.
Then $I$ is said to be a weakly-$w$-ideal if, for all leading monomial $x^\alpha$ of the reduced Gröbner basis of $I,$ regarding to $w,$ for all $x^\beta$ such that $\vert \alpha \vert = \vert \beta \vert$ and $x^\beta > x^\alpha$, $x^\beta$ belongs to $LM(I)$ (according to $w$). 
\end{deftn}

Moreno-Socias has conjectured that this is generic in the following sense:

\begin{conj}[Moreno-Socias]
If $k$ is an infinite field, $s \in \mathbb{N}$, $d_1,\dots, d_s \in \mathbb{N}$, then there is a non-empty Zariski-open subset $U$ in $A_{d_1} \times \dots \times A_{d_s}$ such that for all $(f_1,\dots, f_s) \in U$, $I=(f_1,\dots, f_s)$ is a weakly-grevlex ideal.
\end{conj}

As a consequence, if Moreno-Socias conjecture holds, sequences satisfying \textbf{H1} and \textbf{H2} are generic.
We refer to Moreno-Socias's PhD Thesis \cite{Moreno:1991} or Pardue's article \cite{Pardue:2010} for an introduction to this conjecture.

\begin{rmk}
The choice of grevlex is important: as seen in \cite{Pardue:2010}, if we take 3 quadrics $(f_1,f_2,f_3)$ in $\mathbb{Q}[X_1,\dots, X_6]$, then generically, the ideal $I$ they span is neither lex nor weakly-lex ! 
Indeed, for lex $X_1 > \dots > X_6$, the leading monomials of $I$ in degree 2 are generically $X_1^2,$ $X_1 X_2$ and $X_1 X_3$. 
Yet, in degree 3, we generically have $X_2^3 \in LT(I)$ and $X_1 X_6^2 \notin LT(I)$ while $X_1 X_6^2 > X_2^3$, and $X_2^3$ is not a multiple of any of the leading monomial of $I$ in degree 2, $X_1^2,$ $X_1 X_2$ and $X_1 X_3$. Therefore, the ideal generated by 3 generic quadrics in 6 variables is neither lex nor weakly-lex.
\end{rmk}

\subsection{The weak Matrix-F5 algorithm}

We provide in Algorithm \ref{algo wMF5} the algorithm weak-MF5. We will see in the following subsections that if $(f_1,\dots,f_s)$ is a sequence of homogeneous polynomials in $A=K [X_1, \dots, X_n]$ satisfying \textbf{H1} and \textbf{H2},  and if the $f_i$'s are known up to a large enough precision $O(\pi^k)$ on their coefficients, then weak-MF5 can compute approximate $D$-Gröbner bases of $\left\langle f_1, \dots, f_s \right\rangle.$


\IncMargin{1em}
\begin{algorithm}[H] \label{algo wMF5}
 \DontPrintSemicolon

 \SetKwInOut{Input}{input}\SetKwInOut{Output}{output}
 \SetKwComment{Comment}

 \Input{$F=(f_1, \dots, f_s) \in R[X_1,\dots, X_n]^s$, homogeneous polynomials with respective degrees $d_1 \leq  \dots \leq  d_s$, and $D \in \mathbb{N}$, \\
 a term order $w$.}
 \Output{$(g_1, \dots, g_k) \in A^k$, a D-Gröbner basis of $Id(F)$, or error if $(f_1, \dots, f_s)$ does not satisfy \textbf{H1}, \textbf{H2} or the precision is not enough }

\BlankLine

\Begin{
$G \leftarrow \left\lbrace \right\rbrace$  \;
\For{$d \in \llbracket 0, D \rrbracket$}{
	$\widetilde{\mathscr{M}_{d,0}}:=\emptyset$ \;
	\For{$i \in \llbracket 1,s \rrbracket$}{
		$\mathscr{M}_{d,i}:=\widetilde{\mathscr{M}_{d,i-1}}$ \;
		\For{ $\alpha$ such that $\vert \alpha \vert +d_i=d$}{
			 \If{$x^\alpha$ is not the leading term of a row of $\widetilde{\mathscr{M}_{d-d_i,i-1}}$}{
				Add $x^\alpha f_i$ to $\mathscr{M}_{d,i}$ \;
			}
		}
		Compute $\widetilde{\mathscr{M}_{d,i}}$, the row-echelon form of $\mathscr{M}_{d,i}$, up to the first column with no non-zero pivot \;
		Replace the remaining rows of $\widetilde{\mathscr{M}_{d,i}}$ by multiple of rows of $\widetilde{\mathscr{M}_{d-1,i}}$, so as to obtain an injective matrix in row-echelon form $\widetilde{\mathscr{M}_{d,i}}$. \;
		\eIf{$\widetilde{\mathscr{M}_{d,i}}$ could not be completed}{
		      Return "\textbf{Error,} the ideals are not weakly-$w$, the sequence is not regular, or the precision is not enough ". }
		{ Add to $G$ all the rows of $\widetilde{\mathscr{M}_{d,i}}$ with a new leading monomial. \; } 
	}
}
Return $G$ \;
}
\caption{The weak-MF5 algorithm}
\end{algorithm}
\DecMargin{1em}


\begin{rmk}
At the beginning of the second \textit{for} loop, the classical Matrix-F5 uses $\mathscr{M}_{d,i}:=\widetilde{\mathscr{M}_{d,i-1}}$ instead of $\mathscr{M}_{d,i}:=\mathscr{M}_{d,i-1}.$ The former is faster ($\widetilde{\mathscr{M}_{d,i-1}}$ is already under row-echelon form), but we have chosen the latter since the analysis of the precision is simpler, and more efficient.

\end{rmk}
\begin{rmk}
Instead of adding to $G$ all the rows of $\widetilde{\mathscr{M}_{d,i}}$ with a new leading monomial, it is enough to add the rows whose leading monomial is not a multiple of the leading monomial of a polynomial in $G$, therefore, we will directly obtain at the end a minimal Gröbner base.
\end{rmk}

\subsubsection{Correctness}

We prove here that, regarding to symbolic computation (\textit{i.e.} disregarding precision issues), the  weak-MF5 algorithm indeed compute $D$-Gröbner bases.

\begin{prop} \label{correctness}
Let $(f_1,\dots,f_s) \in B^s$ be a sequence of homogeneous polynomials satisfying \textbf{H1} and \textbf{H2}. Then for any $D \in \mathbb{Z}_{\geq 0},$ the result of weak-MF5($(f_1,\dots,f_s),D$) is a $D$-Gröbner basis of the ideal $I$ generated by $(f_1,\dots,f_s)$.
If $(f_1,\dots,f_s)$ does not satisfy \textbf{H1} or \textbf{H2}, an error is raised.
\end{prop}
\begin{proof}
Let $(f_1, \dots, f_s) \in B^s$, homogeneous of degree $d_1 \leq \dots \leq  d_s$ and satisfying \textbf{H1} and \textbf{H2}. Let $\mathscr{M}_{d,i}$ be the matrix built with the F5-criterion at the beginning of the second \textbf{for} loop in Algorithm \ref{algo wMF5}, and $ \widetilde{\mathscr{M}_{d,i}}$ be the result at the end of this very same loop of $\mathscr{M}_{d,i}$ after row-echelon computation and completion with $ \widetilde{\mathscr{M}_{d-1,i}}.$

Let $\mathcal{P}(d,i)$ be the proposition: $\mathscr{M}_{d,i} = \overline{Mac_d(f_1, \dots, f_i)},$ $\widetilde{\mathscr{M}_{d,i}}$ is put under row-echelon form (up to permutation) by Algorithm \ref{algo wMF5} without raising an error, and $Im(\widetilde{\mathscr{M}_{d,i}}) = Im(Mac_d(f_1,\dots , f_i)).$ We prove by induction on $d$ and $i$ that for any $d \in \llbracket 0, D \rrbracket$ and $i \in \llbracket 1, s \rrbracket,$ $\mathcal{P}(d,i),$ holds.

First of all, $\mathcal{P}(d,i)$ is clear if $d < d_1$ since the corresponding matrices are empty.

Now, let $d \in \llbracket d_1, D \rrbracket$ be such that for any $0 \leq \delta \leq d$ and $i \in \llbracket 1, s \rrbracket,$ $\mathcal{P}(\delta,i)$ is true. We prove $\mathcal{P}(d,i)$ for all $i \in \llbracket 1, s \rrbracket.$

It is clear for $i=1$ since the ideal generated by $f_1$ is monogeneous.
Let $i \in \llbracket 1, s \rrbracket$ be such that for all $j \in \llbracket 1, i-1 \rrbracket,$ $\mathcal{P}(\delta,i)$ is true.

Then, by the induction hypothesis, $\mathscr{M}_{d,i-1} = \overline{Mac_d(f_1, \dots, f_{i-1})},$ $\widetilde{\mathscr{M}_{d-d_i,i-1}}$ is under row-echelon form (up to permutation) and $Im(\widetilde{\mathscr{M}_{d-d_i,i-1}}) = Im(Mac_d(f_1,\dots , f_{i-1})).$ By the F5  criterion (Proposition \ref{discard}), we then indeed have $\mathscr{M}_{d,i} = \overline{Mac_d(f_1, \dots, f_i)},$ $\widetilde{\mathscr{M}_{d,i}}.$
 
Now, we prove that the completion process can be performed without error. Let us denote by $x^{\alpha_u}$, for $u$ from $1$ to $\binom{n+d-1}{n-1}$, the monomials of degree $d$, ordered decreasingly according to $w$, and let $l$ be the index of the first column without pivot found during the computation of the Gaussian row-echelon form of $\mathscr{M}_{d,i} $. 
Let us denote by $r_i$, with $i$ from $1$ to $l-1$, the $l$ polynomials corresponding to rows of $\widetilde{Mac_{d}(f_1, \dots, f_s)}$ with leading monomial $x^{\alpha_i}$.
Their leading monomials belong to the $x^{\alpha_u}$, with $u \geq l$.

Let $(g_1, \dots,g_r)$ be the reduced Gröbner basis of $I$ according to $w$. Then, since there is no pivot on the column of index $l$, $x^{\alpha_l}$ is not a monomial of $LM(I)$. By definition of a weakly-$w$-ideal (hypothesis \textbf{H2}), this implies that if $x^{\alpha_u} \in LM(I)$ for some $u \geq l$, then $x^{\alpha_u}$ is not one of the $LM(g_i)$. This means that any $x^{\alpha_u} \in LM(I)$ with $u \geq l$ is a non-trivial multiple of one of the $LM(g_i)$. Such $x^{\alpha_u}$ is therefore a multiple of a monomial in $LM(I \cap A_{d-1})$.
As a consequence, since $Im(\widetilde{\mathscr{M}_{d-1,i}})=I \cap A_{d-1},$ and $\widetilde{\mathscr{M}_{d-1,i}}$ is under row echelon form (up to permutation), then for any $u\geq l$ such that $x^{\alpha_u} \in LM(I),$ there exists a polynomial $P_u$ corresponding to a row of $\widetilde{\mathscr{M}_{d-1,i}}$ and $k_u \in  \llbracket 1, n \rrbracket$ such that $LM(X_{k_u} P_u)=x^{\alpha_u}.$ With the F5 criterion and \textbf{H1}, $\mathscr{M}_{d,i} = \overline{Mac_d(f_1, \dots, f_i)},$ $\mathscr{M}_{d,i}$ is injective and its number of rows, $m$, is exactly the number of monomials in $LM(I \cap A_d).$ This implies that the completion process can be executed without error.

Let $(t_l,\dots,t_m)$ be the rows of $\widetilde{\mathscr{M}_{d,i}}$ obtained as multiples of rows of $\widetilde{\mathscr{M}_{d-1,i}}.$ Then the polynomials corresponding to $(r_1,\dots, r_{l-1},t_l,\dots,t_m)$ have respectively distinct leading monomials, and therefore, $\widetilde{\mathscr{M}_{d,i}}$ is under row-echelon form (up to permutation). Finally, $Im(\widetilde{\mathscr{M}_{d,i}}) \subset I \cap R_d = Im(Mac_d(f_1,\dots , f_i))$ and both have dimension $m$ over $K.$ Hence, $Im(\widetilde{\mathscr{M}_{d,i}}) = Im(Mac_d(f_1,\dots , f_i)).$

$\mathcal{P}(d,i)$ is then proven. By induction, it is now true for all $d \in \llbracket 0, D \rrbracket$ and $i \in \llbracket 1, s \rrbracket.$ As a consequence, the output of Algorithm \ref{algo wMF5} is indeed a $D$-Gröbner basis of $(f_1,\dots, f_s).$

Now, if $(f_1, \dots, f_s)$ is either not regular or their exist some $i$ such that $(f_1, \dots, f_i)$ is not weakly-$w,$ then in the first case, it means some rows of one of $\overline{Mac_d(f_1, \dots, f_i)}$ reduce to zero  (see \cite{Bardet:2004}). Therefore, the row-echelon computation will encounter a column without pivot and the completion of $\mathscr{M}_{d,i}$ in an echelon basis will not be possible, raising an error. The second case is similar.
\end{proof}

\begin{rmk}
As seen in the proof, the completion process is only here to ensure that no leading monomial is missing and therefore, we indeed have produced a $D$-Gröbner basis. It does not produce new polynomials for the Gröbner basis in development. Had $K$ been an effective field, then under the hypothesis \textbf{H2}, to stop the row-echelon form computation after the first columns without pivot is enough to get the polynomial in a minimal Gröbner basis. One could then ensure that the output is a Gröbner basis by the Buchberger criterion.
\end{rmk}

\subsubsection{Termination}

Since we restrict to computation of row-echelon form of the Macaulay matrices up to degree $D$, there is no termination issue.

Yet, if we want a Gröbner basis instead of a $D$-Gröbner basis, one can use the following result: (see \cite{Bardet:2014}, \cite{Giusti:1984}, \cite{Lazard:1983})

\begin{prop}
If $(f_i, \dots, f_s)$ is a regular sequence of homogeneous polynomials in $A$. Then, after a generic linear change of variables, the highest degree of elements of a Gröbner basis of $<f_1, \dots, f_s>$ for the grevlex ordering is upper-bounded by the Macaulay bound: $\sum_{i=1}^s \left( \vert f_i \vert-1 \right) +1.$
\end{prop}

\subsubsection{Precision}

We can now prove Theorem \ref{thm wMF5}. Let $(f_i, \dots, f_s)$ be a sequence of homogeneous polynomials in $B$ satisfying \textbf{H1} and \textbf{H2}. To that intent, we first define the $\Delta_{d,i},$ which corresponds to the precision sufficient, by Proposition \ref{refined}, to compute the $\widetilde{\mathscr{M}_{d,i}}$ from $\mathscr{M}_{d,i}.$

\begin{deftn}
Let $l_{d,i}$ be the maximum of the $l \in \mathbb{Z}_{\geq 0}$ such that the $l$-first columns of $\overline{Mac_{d}(f_1,\dots, f_i)}=\mathscr{M}_{d,i}$ are linearly free. We define \[\Delta_{d,i}=\min \left( val \left( \left\lbrace \text{minor over the } l_{d,i} \text{-first columns of } \mathscr{M}_{d,i} \right\rbrace \right) \right).\]
\end{deftn}

We can now define $prec_{MF5}$. 

\begin{deftn} \label{precMF5}
We define the Matrix-F5 precision of $(f_1,\dots, f_s)$ regarding to $w$ and $D$ as:
\[prec_{MF5}(\left\lbrace f_1, \dots, f_s \right\rbrace, D, w) = \max_{d\leq D, \: 1\leq i \leq s} val \left( \Delta_{d,i} \right).\]
\end{deftn}

With Proposition \ref{refined} and the Proposition \ref{correctness},  this upper-bound is enough to compute the  $\widetilde{\mathscr{M}_{d,i}}$. Indeed, it is enough to compute the Gaussian row-echelon form of $\mathscr{M}_{d,i}$ up to column $l_{d,i}$ and then to complete this matrix with multiples of $\widetilde{\mathscr{M}_{d-1,i}}.$ This way, either they come from row-reduction or are multiple of rows of $\widetilde{\mathscr{M}_{d-1,i}},$ the leading monomials of the rows of $\widetilde{\mathscr{M}_{d,i}}$ are unambiguous. The fact that the completion process can be successfully completed implies that we are ensured we have obtained an echelon basis for $Im(\mathscr{M}_{d,i}).$

Therefore, $prec_{MF5}(\left\lbrace f_1, \dots, f_s \right\rbrace, D, w)$ is enough to compute approximate $D$-Gröbner bases by the weak-MF5 algorithm.

To conclude the proof, we remark that, in order order to facilitate the precision analysis, we have assumed that the input polynomials in the algorithm, $(f_1, \dots, f_l),$ are in $B.$ Yet, if the $(f_1, \dots, f_l),$ are in $A,$ one can still replace the $f_i$'s by the $\pi^{l_i} f_i \in A$ (for some $l_i$) and still generate the same ideal. This does not affect \textbf{H1} and \textbf{H2}, and one can still compute an approximate Gröbner basis if we know the $f_i$'s up to a large enough precision. Only our precision bound $prec_{MF5}$ is no longer available. 

\subsubsection{Complexity}

Asymptotically, the complexity to compute a $D$-Gröbner basis of $(f_1,\dots, f_s)$  is the same as the classical Matrix-F5 algorithm, that is to say, $O \left( sD \binom{n+D-1}{D}^3 \right)$ operations in $R$ at precision $m,$ as $D \rightarrow +\infty.$
One can see \cite{Bardet:2004} and \cite{Bardet:2014} for more about the complexity of the Matrix-F5 algorithm.

\subsection{Precision versus time-complexity}

In order to achieve a better loss in precision for the Gaussian row-echelon form computation, we suggest the following weak-Matrix algorithm:
\begin{itemize}
\item Compute the $\mathscr{M}_{d,i}$ as before, with the F5 criterion;
\item Instead of computing the row-echelon form of $\mathscr{M}_{d,i}$, one can perform the row-echelon form of the whole $Mac_{d}(f_1,\dots,f_i)$, up to the first column without non-zero pivot;
\item Finally, build $\widetilde{\mathscr{M}_{d,i}}$ by filling $\mathscr{M}_{d,i}$ with the linearly independent rows found by the previous computation over $Mac_{d}(f_1,\dots,f_i)$ and multiples of rows of $\widetilde{Mac_{d-1,i}}$, into a matrix under row-echelon form.
\end{itemize}

The following quantity defines a sufficient precision to compute $D$-Gröbner bases through this algorithm.

\begin{deftn} \label{prec Mac}
Let \[\Box_{d,i} =  \min \left( val \left( \left\lbrace \begin{array}{c}
\text{minor over the } l_{d,i} \text{-first} \\
\text{columns of } Mac_{d}(f_1,\dots,f_i)
\end{array}  \right\rbrace \right) \right). \]

We define the Macaulay precision of $(f_1, \dots, f_s)$ regarding to $w$ and $D$ as:
\[prec_{Mac}(\left\lbrace f_1, \dots, f_s \right\rbrace, D, w) =  \max_{d\leq D, \: 1\leq i \leq s} \Box_{d,i}.\] 
\end{deftn}

Indeed, with Proposition \ref{refined}, $prec_{Mac}(\left\lbrace f_1, \dots, f_s \right\rbrace, D, w)$ is enough to compute approximate $D$-Gröbner bases of sequences of homogeneous polynomials satisfying \textbf{H1} and \textbf{H2}, and it would achieve the best loss in precision that Gaussian row-reduction of Macaulay matrices can attain. We have $prec_{Mac} \leq prec_{MF5}.$

We can illustrate how $prec_{Mac}$ can be strictly smaller than $prec_{MF5}:$ let $f=(5x,y,25xy+z^2)$ in $\mathbb{Q}_5[x,y,z].$ Then $prec_{MF5}(f,2,grevlex(x>y>z))=3$ while  $prec_{Mac}(f,2,grevlex(x>y>z))=2$

Yet, this precision would come with a cost in time-complexity: row-reducing of the full Macaulay matrix $Mac_{d}(f_1,\dots, f_i)$ is $O \left( \binom{n+d-1}{d}^2 \times i \binom{n+d-1}{d} \right)$ (see \cite{Bardet:2014}).
This leads to a total time-complexity in $O \left( s^2 D\binom{n+D-1}{D}^3 \right)$ operations in $R$ at precision $m,$ as $D \rightarrow +\infty,$ while when using the F5-criterion, we only need $O \left( s D\binom{n+D-1}{D}^3 \right).$ To sum up:

\begin{theo} \label{thm wM}
Let $(f_1,\dots,f_s) \in A^s$ be homogeneous polynomials satisfying  \textbf{H1} and \textbf{H2}. Let $(f_1',\dots,f_s')$ be approximations of the $f_i$'s with precision $m$ on the coefficients.
Then, if $m$ is large enough, an approximate $D$-Gröbner basis of $(f_1',\dots,f_s')$ regarding to $w$ is well-defined. It can be computed by the weak Matrix algorithm.

Let $prec_{Mac}=prec_{Mac}(\left\lbrace f_1, \dots, f_s \right\rbrace, D, w).$
Then, if the $f_i$'s are in $B,$ a precision $m \geq prec_{Mac}$ is enough, and the loss in precision is upper-bounded by $prec_{Mac}.$ The complexity is in $O \left( s^2 D \binom{n+D-1}{D}^3 \right)$ operations in $R$ at precision $m,$ as $D \rightarrow +\infty.$
\end{theo}

\section{Topology and optimality} \label{Continuity}

\subsection{Continuity and optimality}

We can reinterpret Theorem \ref{thm wMF5} in the following way: the application $\Phi : A_{d_1} \times \dots \times A_{d_s}\rightarrow \mathcal{P}(A)$ that sends $f=(f_1,\dots,f_s)$ to the set $LM( \left\langle f_1, \dots, f_s \right\rangle )$ (its initial ideal) is locally constant at any sequence satisfying \textbf{H1} and \textbf{H2}. It is these properties that allow numerical stability at $f.$ One could show similarly they still holds for $K= \mathbb{R}.$ Yet in that case, finding an explicit neighborhood of $f$ would be much more involved since we could not apply Proposition \ref{div} and Theorem \ref{Prec-row}.

Concerning the optimality of the structure hypotheses, we remark that without \textbf{H1} or \textbf{H2}, the locally-constant property of $LM$ is not necessarily satisfied.
For example, in $K[X,Y,Z],$ $f=(X+Y,XY+Y^2+Z^2)$ satisfy \textbf{H1} and not \textbf{H2}, and one can consider the approximations $(X+(1+\pi^n)Y,XY+(1-\pi^n)Y^2+Z^2),$ intersecting any neighborhood of $f$ but yielding a different $LM$ than $LM(\left\langle f \right\rangle).$
Likewise, $f=(X+Y,X^2+XY)$ satisfy \textbf{H2} and not \textbf{H1}, with the same issue.

\subsection{Differentiability}

\subsubsection{On differential precision}

In \cite{Caruso:2014:2}, Caruso, Roe and Vaccon have proposed a way to track precision in an ultrametric setting through the usage of differential calculus. It essentially relies on the fact (see Lemma 3.4 of \cite{Caruso:2014:2})  that if $\phi :\: K^n \rightarrow K^m$ is differentiable at some point $x \in K^n$ and its derivative $\phi '(x)$ is surjective, then for any ball $B$ centered in $0$ and small enough, $\phi(x+B)=\phi(x)+\phi'(x) \cdot B.$
The ball $B$ represent the precision at $x,$ and therefore, a consequence of this result is that the precision on $\phi(x)$ defined by $B$ is exactly given by $\phi'(x) \cdot B.$  This allows a very simple and elegant way to handle precision in ultrametric settings: compute the differential of the mapping you are interested in and study its properties.

\subsubsection{Computation of the differential}

To apply this idea of differential precision, we will first compute the differential of the Gröbner bases computation.

Let $f=(f_1,\dots,f_s)$ be a sequence in $A_{d_1} \times \dots \times A_{d_s}$ satisfying \textbf{H1} and \textbf{H2}. Let $U$ be an open neighborhood in $A_{d_1} \times \dots \times A_{d_s}$ of $f$ given for example by $prec_{Mac}(f).$ Let $r$ be the cardinal of the reduced Gröbner basis of $\left\langle  f \right\rangle,$ and $d$ the maximal degree of an element of this basis. We then have the following regularity result:

\begin{prop} \label{prop differentiability}
The application $\Psi$ that maps an element of $U$ to the reduced Gröbner basis of the ideal it generates is a rational, smooth, function $U \rightarrow A_{\leq d}^r.$
\end{prop}

We can now try to apply the idea of \citep{Caruso:2014:2} and its Lemma 3.4. The differential of $\Psi$ is then given by the following result:

\begin{theo} \label{Diff}
Let $f=(f_1,\dots,f_s) \in U.$ Let $g \in A_{\leq d}^r$ be the reduced Gröbner basis of $\left\langle  f \right\rangle,$ and $M \in A^{s \times r}$ be such that $g=f \times M.$
Then \[\Psi'(f) \cdot \delta f = \delta f \times M \mod g.\]
\end{theo}
\begin{proof}
We use the "method of physicists." We have $g+\delta g=(f+\delta f) \times (M+\delta M).$
At the first order, $\delta g = \delta f \times M + f \times \delta M.$
If $\delta f$ is small enough, \textit{e.g.} $\delta f$ is such that $f+\delta f \in U,$ then $g$ and $g+\delta g$ have the same leading terms.
Since $g$ is a reduced Gröbner basis, it implies that none of the $LM(g_i)$'s divides any term of $\delta g.$
Therefore \[\delta g = \delta g \mod g.\]

Besides, $f \times \delta M \in  \left\langle  f \right\rangle.$ Thus, $f \times \delta M = 0 \mod g.$

As a consequence, $\delta g = \delta f \times M \mod g.$
\end{proof}

\begin{rmk}
Even though the computation of $\delta f \times M \mod g$ gives a rather convenient way to track the behavior of Gröbner basis computation arond $f,$ the surjectivity hypothesis in Lemma 3.4 of \citep{Caruso:2014:2} seems rather difficult to apprehend. 
\end{rmk}

\subsubsection{Illustration}

We provide here an example to understand theorem \ref{Diff} and its qualitative implications.
Let $f=(x,xy^2+y^3+z^3)$ in $\mathbb{Q}_p[x,y,z].$ Then a reduced Gröbner basis of $f$ regarding to the grevlex ordering is: $g=(x,y^3+z^3)$ with $M = \begin{bmatrix}
1 & -y^2 \\
0 & 1 \\
\end{bmatrix}.$

Now, let $\delta f = (O(p^5)x,O(p^5)xy^2+O(p^5)y^3+O(p^5)z^3.$
Then, an approximate Gröbner basis of $f+ \delta f$ is given by $(x,y^3+(1+O(p^5))z^3).$ Thus, $\delta g = (0,O(p^5)z^3).$

At the same time, $\delta f \times M = (O(p^5)x,O(p^5)xy^2+O(p^5)y^3+O(p^5)z^3).$ Therefore, $\delta f \times M \mod g = (0,O(p^5)z^3),$ and we do have $\delta f \times M \mod g = \delta g,$ even though the surjectivity hypothesis is not fulfilled.

\section{Implementation} \label{num exp}

\subsection{Direct computations}

A toy implementation in Sage \cite{Sage} of the previous algorithm is available at \url{http://perso.univ-rennes1.fr/tristan.vaccon/toy_F5.py}.

The purpose of this implementation was the study of the precision. It is therefore not optimized regarding to time-complexity.

We have experimented the weak-Matrix-F5 algorithm up to degree $D$ on homogeneous polynomials $f_1,\dots, f_s,$ of degree $d_1, \dots, d_s$ in $\mathbb{Z}_p[X_1,\dots, X_n],$  with coefficients taken randomly in $\mathbb{Z}_p$ up to initial precision $30.$ This experiment is repeated $n_{exp}$ times, and the monomial ordering was grevlex.
\textbf{max} denotes the maximal loss in precision noticed on a polynomial in all the $n_{exp}$ output bases, $\overline{m}$ the mean loss in precision over all coefficients of the output bases, \textbf{gap} is the maximum of the differences for one experiment between the effective maximal loss in precision and the theoretical bound  $prec_{MF5},$ and \textbf{f} is the number of failures.
We present the results in the following array:

\begin{tabular}{|c|c|c|c|c|c|c|c|}
\hline
 $d =$ & D & $p$ & $n_{exp}$ & \textbf{max} & $\overline{m}$ & \textbf{gap} & \textbf{f}  \\  \hline
[3,4,7]& 12 & 2 & 30 &11   & .5 & 141 & 0 \\ \hline
[3,4,7]& 12 & 7 & 30 & 2& 0 & 42 & 0 \\ \hline
[2,3,4,5] & 11 & 2&20 & 25 & 2.2 & 349 & 3 \\ \hline
[2,3,4,5] & 11 & 7&20 & 5 & .3 & 84 & 0 \\ \hline
[2,4,5,6] & 14 & 2&20 & 28 & 3.1 & 581 & 3 \\ \hline
[2,4,5,6] & 14 & 7&20 & 14 & .4 & 73 & 0 \\ \hline
\end{tabular} 
\\

These results suggest that the loss in precision is less when working with bigger primes. It seems reasonable since the loss in precision comes from pivots with positive valuation, while, with the Haar measure on $\mathbb{Z}_p,$ the probability that $val(x) = 0$ for $x \in \mathbb{Z}_p$ is $\frac{p-1}{p}.$ Similarly, it increases when the size of the Macaulay matrices increases.

Concerning the gap between $prec_{MF5}$ and the effective loss in precision, we remark that $prec_{MF5}$ derive from Theorem \ref{Prec-row} which was about dense matrices. A pivot with a big valuation echos into $prec_{MF5},$ even though it might generate no loss in precision if there is no non-zero coefficient on its column to eliminate. Hence, Theorem \ref{Prec-row} does not take into account the sparsity of the Macaulay matrices, which explains why \textbf{gap} might be so big compared to \textbf{max}.

\subsection{About stability}

Thanks to Theorem \ref{Diff}, there are now three ways to experiment the loss in precision of reduced Gröbner bases computation:
\begin{itemize}
\item \textbf{Direct computation}: in $\mathbb{Z}_p[X_1,\dots,X_n],$ through Algorithm \ref{algo wMF5}, starting from some finite precision $O(p^k).$
\item \textbf{Difference method}: Compute the reduced Gröbner bases $g^{(1)}$ of $f^{(1)} \in \mathbb{Z}[X_1,\dots,X_n]^s$ and $g^{(2)}$ of $f^{(2)}=f^{(1)}+\mathrm{d}f,$ with some $\mathrm{d}f \in (p^k \mathbb{Z}[X_1,\dots,X_n])^s,$ and look at the $p$-adic valuations of the coefficients of $g^{(1)}-g^{(2)};$
\item \textbf{Differential}: Compute the differential at $f^{(1)}$ as in Theorem \ref{Diff}.
\end{itemize}

We have tried to compare these three computations over the same context as in the previous Subsection (random homogeneous polynomials of given degrees with initial precision 30). The array in the column Min shows for each of the $n_{exp}$ experiments the minimum of the precision either on the output reduced Gröbner basis or on the estimation of the loss in precision through the difference and differential methods.  

\begin{tabular}{|c|c|c|c|c|c|}
\hline
 $d =$ & D & $p$ & $n_{exp}$ & method & Min  \\  \hline 
[2,2,3]& 5 & 2 & 10 &direct &[28, 24, 6, 15, 25, 27, 22, 22, 26, 22] \\ \hline
[2,2,3]& 5 & 2 & 10 &difference &[30, 25, 6, 18, 25, 30, 28, 24, 27, 22] \\ \hline
[2,2,3]& 5 & 2 & 10 &differential &[30, 25, 6, 18, 25, 30, 28, 24, 27, 22]\\ \hline \hline
[2,2,3]& 5 & 7 & 10 &direct &[29, 26, 28, 29, 29, 30, 26, 30, 28, 30] \\ \hline
[2,2,3]& 5 & 7 & 10 &difference &[30, 26, 29, 30, 29, 30, 27, 30, 28, 30] \\ \hline
[2,2,3]& 5 & 7 & 10 &differential &[30, 26, 29, 30, 29, 30, 26, 30, 28, 30] \\ \hline \hline
[2,3,4]& 7 & 2 & 10 &direct &[22, 26, 28, 27, 24, 23, 27, 18, 21, 22] \\ \hline
[2,3,4]& 7 & 2 & 10 &difference &[23, 26, 29, 28, 26, 29, 28, 20, 26, 22] \\ \hline
[2,3,4]& 7 & 2 & 10 &differential &[23, 26, 29, 28, 26, 29, 28, 20, 26, 22] \\ \hline \hline
[2,3,4]& 7 & 7 & 10 &direct &[30, 28, 28, 30, 28, 30, 28, 28, 26, 30] \\ \hline
[2,3,4]& 7 & 7 & 10 &difference &[30, 29, 29, 30, 28, 30, 28, 28, 26, 30] \\ \hline
[2,3,4]& 7 & 7 & 10 &differential &[30, 28, 28, 30, 28, 30, 28, 28, 26, 30] \\ \hline
\end{tabular} 
\\

From this array, we can infer the following heuristics:
\begin{itemize}
\item Even though the surjectivity hypothesis is not guaranteed, the differential is very close to the results of the difference method, hence to the intrinsic loss in precision;
\item Our direct computations are often stable, but as one could expect, may still occasionally lose more precision than the intrinsic loss in precision.
\end{itemize}

\section{Liftings under the \textbf{H2} condition}

\label{section lift}

\subsection{Lifting Gröbner bases}

In this Section, we address the following issue: in the Definition \ref{App GB} of approximate Gröbner bases, we state that for any specialization of the $O(\pi^{n})$ in the entry polynomials, there is a specialization of the $O(\pi^{n})$ in the output polynomials that stays consistent, namely that is a Gröbner bases of the ideal generated by the entry polynomials. Yet, knowing an approximate Gröbner bases up to some precision, if we give extra digits on the entry polynomials, is there a more clever way to compute an approximate Gröbner basis consistent with those extra-digits than performing a computation from start again?

A first natural idea would be the Hensel-Newton method. The usage of Hensel liftings to speed-up the computation of Gröbner basis have been proposed in \cite{Winkler:1988}, and continued or used in \cite{Pauer:1992}, \cite{Arnold:2003} or \cite{Renault:2006}. The general idea was that, given $f=(f_1,\dots,f_s) \in \mathbb{Z}[X_1,\dots,X_n],$ one would first compute the reduced Gröbner basis $\overline{g}=(\overline{g_1},\dots,\overline{g_r})$ in $\mathbb{Z}/p\mathbb{Z}[X_1,\dots, X_n]$ of the reduction modulo $p$ of $f,$ and then use Hensel liftings to obtain the reduced Gröbner basis $g=(g_1,\dots, g_r)$ in $\mathbb{Z}[X_1,\dots,X_n]$ of $\left\langle f \right\rangle.$
For this scheme to work, one often assume some "luckyness" hypotheses on $p,$ \textit{e.g.} that $\overline{g}$ is the reduction modulo $p$ of $g.$ This conditions are often very difficult to check in advance.
We show in this Section that under the hypotheses \textbf{H1} and \textbf{H2}, one can perform direct lifting of an approximate Gröbner basis (along with the coordinates of its polynomials in the basis of the entry polynomials). The "luckyness" hypothesis of $p$ is then replaced by the only additional requirement of enough entry precision.

\subsection{Lifting at points satisfying \textbf{H1} and \textbf{H2} }

\subsubsection{Presentation of the algorithm}

The idea of the Algorithm \ref{algo lift} is the following. We start from the initial homogeneous polynomials $F \in B^s$ and an approximate $D$-Gröbner basis $G \in B^r$ along with some matrix $M \in B^{s \times r}$ such that $G = (F+O(\pi^m)) \times M,$ and want to compute an approximate $D$-Gröbner basis for $F+O(\pi^l).$ If $F$ can be handled at infinite precision (\textit{e.g.} $F \in \mathbb{Q}[X_1,\dots,X_n]^s$), one can with the same algorithm obtain a $D$-Gröbner basis for $F.$ It corresponds to $l = +\infty$ in what follows. 

The idea is to lift canonically $M$ to some $\widehat{M}$ at precision $O(\pi^l)$ and compute $H=(F+O(\pi^l) \times \widehat{M}.$ The desired approximate $D$-Gröbner basis is then obtained by performing inter-reduction on $H$.

By canonical lifting of $M$ we mean adding zeros as $\pi$-digits up to $O(\pi^l)$ in the $\pi$-adic development of the coefficients of $M.$

\begin{rmk}
In Algorithm \ref{algo lift}, we use Algorithm \ref{algo wMF5} with the small modification that it computes an $M$ such that $G=F \times M$ along with the computation of $G.$ 
\end{rmk}   

\IncMargin{1em}
\begin{algorithm}[H] \label{algo lift}
\DontPrintSemicolon

\SetKwInOut{Input}{input}\SetKwInOut{Output}{output}
\SetKwComment{Comment}{}{}

\Input{$F=(f_1, \dots, f_s) \in R[X_1,\dots, X_n]^s$, homogeneous polynomials with respective degrees $d_1 \leq  \dots \leq  d_s$, and $D \in \mathbb{N}$, \\
 a term order $w$. \\
 The precision $m$ of the first computation and the precision $l$ to perform the lifting. \\}
\Output{$(g_1, \dots, g_k) \in A^k$, an approximate D-Gröbner basis of $Id(F)$, or error if $(f_1, \dots, f_s)$ does not satisfy \textbf{H1}, \textbf{H2} or the precision is not enough. }

\BlankLine

\Begin{
$G,M \leftarrow \textrm{weak-MF5}(F+O(\pi^m));$ \; 
\Comment*{We have $G=(F+O(\pi^m)) \cdot M$;}  \;
$\widehat{M} \leftarrow $ canonical lifting of $M$ to precision $l;$ \;
$H \leftarrow (F+O(\pi^l)) \cdot \widehat{M};$ \;
$\widehat{G} \leftarrow [];$ \;
\For{$i \in \llbracket 1, \sharp H $ }{
$\widehat{G} .\textrm{Append}(H[i] \mod \widehat{G});$ \;
}
Return $\widehat{G}.$\;
}
\caption{The weak-lifting algorithm}
\end{algorithm}
\DecMargin{1em}

\subsubsection{Correction}

We prove here that Algorithm \ref{algo lift} do compute a lifted approximate $D$-Gröbner basis.

Indeed, let $f=(f_1,\dots,f_s) \in B_{d_1}\times \dots B_{d_s} $ satisfying \textbf{H1} and \textbf{H2}. Let $n$ be a high-enough precision on $f.$ Let $G=(g_1,\dots,g_r) \in B^r$ be some approximate $D$-Gröbner basis of $\left\langle f+O(\pi^m) \right\rangle,$  with $M \in B^{s \times r}$ homogeneous polynomials, such that for some $l_1 >0,$ we have $g+O(\pi^{l_1})=(f+O(\pi^m)) \times (M+\pi^{m}).$ We assume that the $g_i$'s are sorted such that $\vert g_i \vert <  \vert g_{i+1} \vert,$ or $\vert g_i \vert =  \vert g_{i+1} \vert$ and $LM(g_i) > LM(g_{i+1}).$ The output of Algorithm \ref{algo wMF5} weak-MF5 satisfies this condition.

Let $l>m,$ or $l=+\infty$ in the special case that infinite precision can be given on $f.$ 
Let $\widehat{M}$ be the trivial lifting of $M+O(\pi^{m})$ to $O(\pi^l),$ \textit{i.e.} $\widehat{M}=M+O(\pi^{m})$ and the $\pi$-digits of $\widehat{M}$ between $\pi^{m}$ and $\pi^l$ are zeros. The following Lemma proves the correction of Algorithm \ref{algo lift}:

\begin{lem} \label{lem lift}
Let $H=(f+O(\pi^m)) \times (\widetilde{M}+O(\pi^{m}).$ Then, for $m$ and $l_1$ big enough, $\widehat{G}$ given by the successive reduction of the $h_i$'s by the $(\widehat{g}_1,\dots,\widehat{g}_{i-1})$ is an approximate $D$-Gröbner basis of  $\left\langle f+O(\pi^m) \right\rangle.$
If $G$ is the output of Algorithm \ref{algo wMF5} then $n>2 prec_{MF5}(f)$ is enough and $l_1=n-prec_{MF5}(f).$
If $G$ is a reduced Gröbner basis up to renormalisation of its leading coefficients, then so is $\widehat{G}.$
\end{lem}
\begin{proof}
We first remark that $H=G +O(\pi^{l_1}),$ and $H$ is made of homogeneous polynomials because so are $\widehat{M}$ and $F.$ 
We then prove by induction that for all $i,$ $LM(\widehat{g_i})=LM(g_i).$

Because of the $H2$ hypothesis, $LM(g_1)$ is the biggest monomial of degree $\vert g_1 \vert.$ Since $h_1=\widehat{g_1}$ is homogeneous of the same degree as $g_1$ and  $\widehat{g_1}=g_1 +O(\pi^{l_1}),$ $LM(\widehat{g_1})=LM(g_1).$
Let us assume that for some $i>0,$ we have that for any $1\leq j <i,$ $LM(\widehat{g_j})=LM(g_j).$ Let $x^\beta = LM(g_i).$
Then, the $\textbf{H2}$ hypothesis means that for any $x^\alpha >x^\beta$ and of degree $\vert g_i \vert,$ there is some $LM(\widehat{g}_j)$ with $j<i$ that divides it.
Therefore, $LM(\widehat{g}_i) \leq LM(g_i).$
Nevertheless, the coefficients in $x^\alpha$ of $h_i,$ with $x^\alpha >x^\beta$ and of degree $\vert g_i \vert,$ belong to $\pi^{l_1}R.$ As a consequence, there exists some $c \geq 0$ such that after the reduction of $h_i$ by the $\widehat{g_j},$ the coefficients in $x^\beta$ of $g_i$ and $\widehat{g_i}$ are equal modulo $\pi^{l_1-c}.$ Hence, if $l_1$ is big enough, it implies that  this coefficient is non-zero and then $LM(\widehat{g}_i) = LM(g_i).$

In the special case that $G$ is the output of Algorithm \ref{algo wMF5} and $n>2 prec_{MF5}(f),$ $l_1=m-prec_{MF5}(f),$  then $c \leq prec_{MF5}$ and $l_1$ is indeed big enough.
The result is then proved.

Regarding reduced Gröbner bases, with $LM(\widehat{g}_i) = LM(g_i)$ and the definition of $LM(\widehat{g}_i),$ the result is clear.
\end{proof}

Heuristically, $n$ and $l_1$ are big enough when $G$ can take a "second" row-reduction to certify its leading monomials. This is why when we first apply Algorithm \ref{algo wMF5}, $n > 2 prec_{MF5}$ is enough. 

\subsubsection{About complexity}

The complexity of Algorithm \ref{algo lift} is the following. We can remark that it implies that the cost of the linear algebra is completely borne by the arithmetic at the initial finite precision, $m$.

\begin{prop} \label{prop complexite lift}
The complexity of Algorithm \ref{algo lift} is in $O \left( s^2 D \binom{n+D-1}{D}^3 \right)$ operations in $K$ at precision $m$ and $O \left( (s+\sharp G ) \binom{n+D-1}{D}^2 \right)$ operations at precision $l.$
\end{prop}
\begin{proof}
The computation of $M$ along that of $g$ adds a factor $s$ to the asymptotic complexity of the computation of a $D$-Gröbner basis, namely, it accounts to a total complexity in $O \left( s^2 D \binom{n+D-1}{D}^3 \right)$ operations in $K$ at the initial precision $m.$ Indeed, it is enough to extend the Macaulay matrices with labels expressing the row polynomials in the basis of the initial polynomials and reflect the operations on the rows on these labels. This adds a factor $s$ to the computation of operations on the rows.  
The complexity of the computation of $H$ is then in $O \left( s \binom{n+D-1}{D}^2 \right)$ operations at precision $l,$ and the reductions to compute $\widehat{G}$ is in $O \left( \sharp G  \binom{n+D-1}{D}^2 \right)$operations at precision $l.$
\end{proof}

\subsection{Application}

We remark that in the Lemma \ref{lem lift}, one can lift the entries to precision $+\infty.$ It implies that if the input polynomials $f=(f_1,\dots,f_s)$ have coefficients in $\mathbb{Q},$ it is then possible to first compute an approximate Gröbner basis of $f$ to some high-enough precision, and then compute a Gröbner basis of $f$ as in Lemma \ref{lem lift}. Thanks to the fact that linear algebra is performed at finite precision, it amounts to a total  complexity that can be seen as an intermediate between that of computation of Gröbner bases over finite fields and direct computation over $\mathbb{Q}.$

This result can then be an answer to the issue with the difficulty of finding "lucky" primes. Taking $p$ of "medium size," like $7$ and working at a moderate precision, like $11$, would have been enough to compute and lift Gröbner bases in all the 20 cases tested in the case $d=[2,3,4,5]$ in Section \ref{num exp}. A precision $30$ would have been enough in the 20 cases tested of $d=[2,4,5,6],$ and in most cases, fewer precision was needed. 
One can also chose to take $p$ of bigger size, which allows to work at smaller precision. 

Moreover, because of the continuity result in Proposition \ref{prop differentiability}, while being cautious with the precision, the call to the weak-MF5 algorithm in Algorithm \ref{algo lift} could be replaced by any computation of Gröbner bases (\textit{e.g.} Faugère's F4 or F5), to obtain a better complexity for the computation of the first approximate Gröbner basis.

Finally, we illustrate our strategy with a direct example. Let $f=(10x,25xy^2+y^3+z^3)$ in $\mathbb{Q}[x,y,z].$ We first work with $\widetilde{f}=(10+O(5^4))x,(25+O(5^4))xy^2+(1+O(5^4))y^3+(1+O(5^4))z^3)$ in $\mathbb{Q}_5[x,y,z]$ at initial precision $4.$
Then Algorithm \ref{algo wMF5} provides, after elimination of rows with non-minimal leading monomial, the following approximate Gröbner basis: $G=((10+O(5^4))x,(1+O(5^3))y^3+(1+O(5^3))z^3),$ along with $M = \begin{bmatrix}
1 & -(3*5 + 2*5^2+O(5^3) )y^2 \\
0 & 1 \\
\end{bmatrix}.$

This leads to $\widehat{M} = \begin{bmatrix}
1 & -65y^2 \\
0 & 1 \\
\end{bmatrix},$ which, along with $f,$ gives $H=(10x,-5^4 xy^2+y^3+z^3).$ Inter-reduction then provides the minimal (and reduced up to the leading coefficients)   $\widehat{G} = (10x,y^3+z^3).$

\section{The affine case}

It is possible to extend Theorem \ref{thm wMF5} to non-homogeneous entry polynomials in a fashion similar to \cite{Faugere:2013} or \cite{Faugere:2014} when the monomial ordering $w$ refines the total-degree order.

\begin{prop} \label{thm wMF5 affine}
Let $(f_1,\dots,f_s) \in K[X_1,\dots, X_n]^s$ be polynomials. Let $f_1^h,\dots,f_s^h $ be their homogeneous components of highest degree. We assume that  $(f_1^h,\dots,f_s^h )$ satisfies \textbf{H1} and \textbf{H2}. Let $(f_1',\dots,f_s')$ be approximations of the $f_i$'s with precision $m$ on the coefficients.
Then, if $m$ is large enough, an approximate $D$-Gröbner basis of $(f_1',\dots,f_s')$ regarding to $w$ is well-defined. 

Moreover, if the $f_i$'s are in $R[X_1,\dots,X_n],$ then $m \geq prec_{MF5}(f_1^h,\dots,f_s^h )$ is enough, and the loss in precision is upper-bounded by $prec_{MF5}(f_1^h,\dots,f_s^h )$ or $prec_{Mac}(f_1^h,\dots,f_s^h )$ 
\end{prop}
\begin{proof}
As in proposition 13 of \cite{Faugere:2014}, $LM( \left\langle f_1^h,\dots,f_s^h \right\rangle)=LM( \left\langle f_1,\dots,f_s \right\rangle).$
Let $(h_1,\dots,h_r) \in K[X_1,\dots, X_n]^r$ be a Gröbner basis of $\left\langle f_1^h,\dots,f_s^h \right\rangle,$ made of homogeneous polynomials $h_i = \sum_{j} a_{i,j} f_j^h,$ for some $a_{i,j} \in K[X_1,\dots, X_n],$ homogeneous, and $r\geq 0.$ 
Let $(g_1,\dots,g_r) \in K[X_1,\dots, X_n]^r$ be such that $g_i=\sum_{j} a_{i,j} f_j.$ 
Then as a consequence, $(g_1,\dots,g_r)$ is a Gröbner basis of $ \left\langle f_1,\dots,f_s \right\rangle.$

It therefore means that the computation of a Gröbner basis of $ \left\langle f_1,\dots,f_s \right\rangle$ can be completely determined by that of $\left\langle f_1^h,\dots,f_s^h \right\rangle.$ It is then clear that it suffice to apply Theorem \ref{thm wMF5} or Theorem \ref{thm wM} to $(f_1^h,\dots,f_s^h).$
\end{proof}

\begin{rmk}
One can define an affine Moreno-Socias conjecture : If $k$ is an infinite field, $s \in \mathbb{N}$, $d_1,\dots, d_s \in \mathbb{N}$, then there is a non-empty Zariski-open subset $U$ in $A_{ \leq d_1} \times \dots \times A_{ \leq d_s}$ such that for all $(f_1,\dots, f_s) \in U$, $I=(f_1,\dots, f_s)$ is a weakly-grevlex ideal.
The affine and non-affine conjecture are clearly equivalent. As a consequence, Moreno-Socias conjecture would imply that sequences $(f_1,\dots, f_s) \in A_{ \leq d_1} \times \dots \times A_{ \leq d_s}$ satisfying \textbf{H1} and \textbf{H2} are generic.
\end{rmk}

\bibliographystyle{alpha}
\bibliography{biblio}

\end{document}